\documentclass[aps,pre,11pt,nofootinbib,byrevtex,showkeys,longbibliography]{revtex4-1}
\usepackage{microtype}
\usepackage{amsmath}
\usepackage{amsfonts}
\usepackage{amssymb}
\usepackage{amsthm}
\usepackage{graphics}
\usepackage[colorlinks,linkcolor=blue,urlcolor=blue,citecolor=blue]{hyperref}

\newcommand{\be}{\begin{equation}}
\newcommand{\ee}{\end{equation}}
\newcommand{\om}{\omega}

\newcommand{\ra}{\rightarrow}

\newcommand{\cU}{\mathcal{U}}
\newcommand{\reals}{\mathbb{R}}
\newcommand{\dom}{\operatorname{dom}}
\newcommand{\rh}{\tilde h}

\newcommand{\cE}{\mathcal{E}}
\newcommand{\cM}{\mathcal{M}}

\newcommand{\cX}{\mathcal{X}}

\newcommand{\id}{\mathbf{1}}

\newtheoremstyle{myplain}
{5pt}			
{5pt}			
{\itshape}	
{}			
{\bfseries}		
{.}			
{.5em}		
{\thmname{#1}\thmnumber{ #2}\thmnote{~{(#3)}}}

\theoremstyle{myplain}
\newtheorem{theorem}{Theorem}

\newtheorem{definition}[theorem]{Definition}

\newtheorem*{assumption}{Assumptions}

\newtheorem{corollary}[theorem]{Corollary}
\newtheorem{proposition}[theorem]{Proposition}

\begin{document}

\title{Equivalence and nonequivalence of ensembles:\\ Thermodynamic, macrostate, and measure levels}

\author{Hugo Touchette}
\email{htouchette@sun.ac.za}
\affiliation{National Institute for Theoretical Physics (NITheP), Stellenbosch 7600, South Africa}
\affiliation{Department of Physics, Stellenbosch University, Stellenbosch 7600, South Africa}
\affiliation{Institute of Theoretical Physics, Stellenbosch University, Stellenbosch 7600, South Africa}

\date{\today}

\begin{abstract}
We present general and rigorous results showing that the microcanonical and canonical ensembles are equivalent at all three levels of description considered in statistical mechanics~--~namely, thermodynamics, equilibrium macrostates, and microstate measures~--~whenever the microcanonical entropy is concave as a function of the energy density in the thermodynamic limit. This is proved for any classical many-particle systems for which thermodynamic functions and equilibrium macrostates exist and are defined via large deviation principles, generalizing many previous results obtained for specific classes of systems and observables. Similar results hold for other dual ensembles, such as the canonical and grand-canonical ensembles, in addition to trajectory or path ensembles describing nonequilibrium systems driven in steady states. 
\end{abstract}


\maketitle

\tableofcontents

\newpage

\section{Introduction}

The problem of determining whether the microcanonical and canonical ensembles give the same predictions has a long history in statistical mechanics, starting from Boltzmann's introduction of these ensembles as the \textit{ergode} and \textit{holode} \cite{boltzmann1968}, respectively, and Gibbs's formulation of these ensembles in their modern probabilistic form \cite{gibbs1902}. Depending on the level of description considered, this equivalence problem takes different forms:
\begin{itemize}
\item \textbf{Thermodynamic equivalence:} Are the microcanonical thermodynamic properties of a system determined from the entropy as a function of energy the same as the canonical thermodynamic properties determined from the free energy as function of temperature? Are energy and temperature always one-to-one related?

\item \textbf{Macrostate equivalence:} Is the set of equilibrium values of macrostates (e.g., magnetization, energy, velocity distribution, etc.) determined in the microcanonical ensemble the same as the set of equilibrium values determined in the canonical ensemble? What is the general relationship between these two sets?

\item \textbf{Measure equivalence:} Does the Gibbs distribution defining the canonical ensemble at the microstate level converge (in some sense to be made precise) to the microcanonical distribution defined by Boltzmann's equiprobability postulate?

\end{itemize}

Many results have been derived over the years, providing conditions for equivalence at each of these levels, as well as conditions for relating one level to another; see \cite{touchette2004b} for a review. It is known in particular that equivalence holds at the thermodynamic level whenever the entropy is concave and that this  also implies, under additional conditions, the equivalence of the microcanonical and canonical ensembles at the macrostate level. Although the first result is general -- it is just a mathematical statement about concave functions and the duality of Legendre transforms -- the second has been derived by Ellis, Haven and Turkington \cite{ellis2000} for a class of systems comprising mostly ideal (non-interacting), long-range, and mean-field systems. As for the measure level, a number of general results have been obtained by Lewis, Pfister and Sullivan \cite{lewis1994a,lewis1994,lewis1995}, but have not been completely related to the two other levels for general systems and macrostates. 

The aim of this paper is to survey these results and to complete them by proving in a general way that equivalence holds at each of the level above under the same condition, namely the concavity of the entropy. The main ideas behind our results were presented in \cite{touchette2011b}; here we focus on giving rigorous proofs, as well as on studying the measure level, which is not considered in \cite{touchette2011b}. For the macrostate level, our results significantly generalize those of Ellis, Haven and Turkington \cite{ellis2000} to any macrostate of any classical many-particle system for which equilibrium statistical mechanics is defined. The same results also imply equivalence results for the measure level, generalizing those of Lewis, Pfister and Sullivan \cite{lewis1994a,lewis1994,lewis1995} in terms of systems and observables considered, and the way equivalence at this level is defined. In the end, our results show that all three equivalence levels coincide, under the condition that entropy be concave.

This condition is important because recent studies have shown that many physical systems have nonconcave entropies in the thermodynamic limit. The common property of these systems is that they involve long-range interactions that asymptotically decay at large distances $r$ according to $r^{-\alpha}$ with $0\leq \alpha\leq d$, where $d$ is the dimensionality of the system. Thus, the dividing line between equivalence and nonequivalence of ensembles is essentially between short- and long-range systems: the former have concave entropies, and are thus described equivalently by the microcanonical or the canonical ensemble, as proved by Ruelle \cite{ruelle1969} (see also \cite{griffiths1972}), whereas the latter can have nonconcave entropies and therefore nonequivalent ensembles.\footnote{Although this has not been proved rigorously, it is thought that the presence of long-range interactions is a necessary but not sufficient condition for having nonconcave entropies. It is known at least that not all long-range systems have nonconcave entropies.} Gravitating particles are historically and physically the most important example of long-range systems showing this behavior, as discovered by Lynden-Bell \cite{lynden1968,lynden1977,lynden1999} and Thirring  \cite{thirring1970} in the late 1960s, and as extensively studied since then; see \cite{chavanis2006,campa2009,dauxois2010} for recent reviews. Other examples include non-screened plasma, dipolar systems, statistical models of two-dimensional turbulence, and mean-field systems in general; see \cite{campa2009} for a review. Recently, experiments based on ion, cold atom and optical traps have  been proposed to observe long-range interactions and nonconcave entropies \cite{kastner2010,kastner2010b,olivier2014,chalony2013}.

The recent study of these long-range systems explains the need to revisit the equivalence problem. Indeed, most works on this problem assume either implicitly or explicitly that entropy is always concave, and so conclude that ensembles are always equivalent (except possibly at phase transitions, as already noted by Gibbs \cite{gibbs1902}); see, for example, \cite{gurarie2007,galgani1970,galgani1971}. In Ruelle's work \cite{ruelle1969}, equivalence is not assumed but follows directly from the class of interactions considered, namely short-range and tempered, for which it can be proved using subadditivity arguments that the entropy exists and is concave in the thermodynamic limit (see also \cite{lanford1973}). The same applies to more recent works on Gibbs states and ensemble equivalence at the level of the so-called \emph{empirical process} or \emph{level-3} macrostate; see \cite{deuschel1991,stroock1991,stroock1991b,roelly1993} and, in particular, the work of Georgii \cite{georgii1993,georgii1994,georgii1995}.

Long-range systems can have nonconcave entropies precisely because the subadditivity argument is not applicable: in the presence of long-range interactions, one cannot divide a system into subsystems in such a way that the total energy of the system is extensive in the energies of the subsystems \cite{campa2009}. The thermodynamic limit of this system can still be defined using Kac's rescaling prescription \cite{kac1963} for all the usual thermodynamic quantities (entropy, free energy, etc.), so that statistical mechanics applies to long-range systems in the same way as for short-range systems \cite{campa2009}. The difference, however, is that the entropy function is not necessarily concave.

Here, we investigate the consequences of this property for the equivalence of the microcanonical and canonical ensembles in the case of general classical $N$-particle systems. Unlike several works on the subject, we do not consider specific systems defined by a class of Hamiltonians, but rather assume that the Hamiltonian is given and that the thermodynamic potentials and equilibrium states obtained from this Hamiltonian exist in each ensemble and are characterized, as explained in the next section, by well-defined large deviation principles. This is a natural assumption given, on the one hand, that the equivalence problem has no meaning when thermodynamic potentials and equilibrium states do not exist and, on the other, that all cases of thermodynamic behavior and equilibrium states known to date are described by large deviation theory.\footnote{There are strong reasons to believe that this cannot be otherwise: that is, many-body systems should have equilibrium states in the thermodynamic limit only when they are described by large deviation theory or, more precisely, when their distribution follows what is called the large deviation principle; see Sec.~\ref{secdef}.} With this assumption, we then prove that ensemble equivalence holds at the thermodynamic, macrostate and measure levels when the entropy is concave in the thermodynamic limit. This generalizes and unifies, as mentioned, all previous results on this problem.

For simplicity, we focus in this paper on the microcanonical and canonical ensembles, but as mentioned in Sec.~\ref{secother} the results proved also hold with minor modifications to other dual ensembles, such as the canonical and grand-canonical ensembles, the volume and pressure ensembles, and the magnetization and magnetic field ensembles. In each case, the entropy function entering in the equivalence condition has to be replaced by the thermodynamic potential of the constrained ensemble considered, for example, the entropy as a function of the particle density for the canonical ensemble. As  shown in that section, the same notion of equivalence also applies to nonequilibrium generalizations of the canonical and microcanonical ensembles defined for paths of Markov processes. What underlies the problem of ensemble equivalence is in fact a general relationship between the conditioning and the so-called tilting of probability measures \cite{lewis1994a,lewis1994,lewis1995}, arising in many problems in probability theory, stochastic simulations, and the study of stochastic processes.

The organization of this paper is as follows. In Sec.~\ref{secdef} we define the microcanonical and canonical ensembles, as well as the basic large deviation principles used to define the set of equilibrium macrostates in each ensemble. This section follows the standard construction of these ensembles in terms of large deviations, which can be found for example in \cite{ellis1985,ellis1995,ellis1999,touchette2009}. In Secs.~\ref{secthermo}, we state some known definitions and results about thermodynamic equivalence, and then prove in Secs.~\ref{secmacro} and \ref{secmeas} new equivalence results that relate this level to the macrostate and measure levels, respectively. The main insight used for proving equivalence at the macrostate level is an exact variational principle, stated in Sec.~\ref{secmacro}, relating the typical states and fluctuations of the microcanonical ensemble to those of the canonical ensemble. Rigorous proofs of all the results follow using a combination of convex analysis results, summarized in Appendix~\ref{appconc}, and a fundamental result of large deviation theory known as Varadhan's Theorem, stated in Appendix~\ref{appvar}. Finally, in Sec.~\ref{secother}, we show how to generalize our results to other dual equilibrium and nonequilibrium ensembles, as mentioned above. 

\section{Notations and definitions}
\label{secdef}

We introduce in this section the notations used in the paper and recall the definitions of the microcanonical and canonical ensembles following the large deviation theory approach to statistical mechanics~\cite{ellis1985,ellis1995,ellis1999,touchette2009}. The notations closely follow those of \cite{ellis2000}. 

\subsection{Systems and macrostates}

We consider a system of $N$ classical particles with microscopic configuration or microstate $\om=(\om_1,\om_2,\ldots,\om_N)\in \Lambda_N=\Lambda^N$, where $\om_i$ is the state of the $i$th particle taking values in some space $\Lambda$. The total energy of the system is given by its Hamiltonian $H_N(\om):\Lambda_N\ra\reals$, from which we define the mean energy or energy per particle as $h_N(\om)=H(\om)/N$. For simplicity, we do not consider the volume of the system, so that the thermodynamic limit is obtained by taking the limit $N\ra\infty$ with $h_N$ kept constant. Systems with a volume are considered  in Sec.~\ref{secother}, which treats the equivalence of the canonical and grand-canonical ensembles.

At the macroscopic level, the $N$-particle system is characterized in terms of a macrostate, defined mathematically as a function $M_N:\Lambda_N\ra\cM$ taking values in some measurable space $\cM$. This macrostate can represent, for example, the mean magnetization of a spin system, in which case $\cM=[-1,1]$, or the empirical distribution of velocities or positions of a gas of $N$ particles, in which case $\cM$ is the space of normalized probability distributions. Note that the same symbol $M_N$ is used to denote a single (scalar) macrostate or a sequence (vector) of macrostates.

To treat the general case where $M_{N}$ can take values in a function space, $\cM$ is usually considered in large deviation theory to be a topological space known as a \emph{Polish space}, which is a metric, separable and complete topological space \cite{dembo1998}.\footnote{There are many reasons for considering Polish spaces: one is that projections of measurable subsets of a Polish space are measurable; another is that the set of probability measures defined on a Polish space is also Polish; see Appendix D of \cite{dembo1998} for more details.} In this paper, we follow a more practical approach and consider $\cM$ to be a subset of $\reals^d$ with the usual Euclidean metric. This is not a fundamental restriction, since our main results rely, as will be noted, on general large deviation results stated in the context of Polish spaces, but it is convenient to simplify the notations and to avoid unnecessary abstract topological issues. Empirical distributions and other similar macrostates defined over function spaces can be treated in $\reals^{d}$ by discretizing them into finite-dimensional vectors and by taking, as is standard in physics, the continuum limit. In large deviation theory, this discretization is handled rigorously with the concept of projective limits; see Sec. 4.6 of \cite{dembo1998}. 

To construct a statistical description of the $N$-particle system, we finally need a prior measure $P_N$ on $\Lambda_N$, whose basic element is denoted either by $P_N(d\om)$ or $dP_N(\om)$. In statistical physics, this prior is almost always taken to be the (non-normalized) Lebesgue measure $d\om$. Here, we follow \cite{ellis2000} and make $P_N$ explicit in the definition of statistical ensembles. This has the advantage of allowing one to consider models which do not necessarily comply with Boltzmann's equiprobability postulate, and so for which the prior is not necessarily proportional to $d\om$. For a comparison of the two approaches, see Secs.~5.1 and 5.2 of \cite{touchette2009}.

\subsection{Equilibrium ensembles} 

We consider throughout most of the paper the equilibrium properties of macrostates as calculated in the canonical and microcanonical ensembles. Generalizations of the results obtained for other ensembles are presented in Sec.~\ref{secother}.

The canonical ensemble is defined in the usual way by the microstate probability measure
\be
P_{N,\beta}(d\om)=\frac{e^{-\beta H_N(\om)}}{Z_N(\beta)} P_N(d\om)
\label{eqcan1}
\ee
where
\be
Z_N(\beta)=E_{P_N}[e^{-\beta H_N(\om)}]=\int_{\Lambda_N}e^{-\beta H_N(\om)}\, dP_N(\om)
\label{eqzn1}
\ee
is the \emph{partition function} normalizing $P_{N,\beta}$ and $\beta=(k_BT)^{-1}$ is the inverse temperature. From this measure, the probability of any event $M_N\in A$ involving a macrostate $M_N$ and some measurable subset $A$ of $\cM$ is calculated as
\be
P_{N,\beta}\{M_N\in A\}=\int_{\Lambda_N} \id_A \big(M_N(\om)\big)\ P_{N,\beta}(d\om)=\int_{M^{-1}_N(A)} dP_{N,\beta}(\om)
\ee
where
\be
M_N^{-1}(A)=\{\om\in\Lambda_N: M_N(\om)\in A\}
\ee
is the preimage of $M_N\in A$.

The microcanonical ensemble is defined, on the other hand, via Boltzmann's equiprobability postulate by assigning a constant weight to all microstates $\om$ having an energy $H_N(\om)=U$. This is generalized in the case of a general prior $P_{N}$ by conditioning $P_N$ on the set of microstates having a mean energy $h_N(\om)$ lying in the `thickened' energy shell $[u-r,u+r]$:
\be
P^{u,r}_N(d\om) = P_N\{d\om| h_N\in [u-r,u+r]\},\qquad r>0.
\ee
By Bayes' Theorem, this becomes
\be
P^{u,r}_N(d\om)
= \frac{P_{N}\{h_{N}\in [u-r,u+r]|\om\} P_{N}(d\om)}{P\{h_{N}\in[u-r,u+r]\}}
=\frac{\id_{[u-r,u+r]}\big(h_N(\om)\big)}{P_N\{h_N\in [u-r,u+r]\}}P_N(d\om),
\ee
where $\id_A(x)$ is the indicator function of the set $A$ and
\be
P_N\{h_N\in [u-r,u+r]\}=\int_{\Lambda_N} \id_{[u-r,u+r]}\big(h_N(\om)\big)\, dP_N(\om)
\ee 
is a normalizing factor representing the mean energy distribution with respect to the prior $P_N$. The need to consider the mean energy rather than the energy in the definition of $P_N^{u,r}$ arises because of the thermodynamic limit, whereas the thickened energy shell is there to make $P^{u,r}_N$ a well defined probability measure. 

Physically, the limit $r\ra 0$ must of course be taken to obtain results that are independent of $r$ \cite{ellis2000,pathria1996,dorlas1999}. From now on, this limit will be implicit, so we omit $r$ in $P^{u,r}_N$ and replace the interval $[u-r,u+r]$ by the infinitesimal element $du$, so as to write the microcanonical ensemble at mean energy $u$ simply as
\be
P^u_N(d\om)=P_N\{d\om| h_N\in du\}=\frac{\id_{du}\big(h_N(\om)\big)}{P_N\{h_N\in du\}} P_N(d\om).
\label{eqmicro1}
\ee
From this microstate measure, macrostates probabilities at fixed energy are then calculated, as in the canonical case, using
\be
P_N^u\{M_N\in A\}=\int_{\Lambda_N} \id_A \big(M_N(\om)\big)\ P_N^u(d\om)=\int_{M_N^{-1}(A)} dP_N^u(\om). 
\ee
A more physical but less rigorous approach based on probability densities instead of probability measures, which avoids the use of $r$, can be found in \cite{touchette2011b}; alternatively, see \cite{lewis1995} for a definition of the microcanonical ensemble based on a set conditioning of the form $h_N\in A_N$.

\subsection{Large deviation principles}

The stability of equilibrium systems observed physically at the macroscale arises because large fluctuations of macrostates are extremely unlikely due to the fact that $P_{N,\beta}$ and $P_N^u$ concentrate exponentially with the system size around certain values of $M_{N}$ corresponding to equilibrium states. This exponential concentration is known in probability theory as the \emph{large deviation principle} (LDP) and is defined as follows. Consider first the canonical ensemble. We say that \emph{$M_N$ satisfies the LDP with respect to $P_{N,\beta}$ if there exists a lower semicontinuous function $I_\beta$ such that}
\be
\limsup_{N\ra\infty} \frac{1}{N}\ln P_{N,\beta}\{M_N\in C\}\leq - \inf_{m\in C} I_\beta(m)
\label{eqldpb1}
\ee
\emph{for any closed sets $C$ and}
\be
\liminf_{N\ra\infty} \frac{1}{N}\ln P_{N,\beta}\{M_N\in O\}\geq -\inf_{m\in O} I_\beta(m)
\label{eqldpb2}
\ee
\emph{for any open sets $O$}.\footnote{We should define the LDP more precisely for the \emph{sequence} $\{P_{N,\beta}\}$ of probability measures associated with the \emph{sequence} $\{M_N\}$ of random variables. Here, we simplify the presentation by referring directly to macrostates and their probabilities.} The function $I_\beta$ is called the \emph{rate function}; in \cite{lewis1994a,lewis1994,lewis1995}, it is also called the Ruelle-Lanford (R-L) function. In the microcanonical ensemble, we say similarly that $M_N$ satisfies the LDP with respect to $P_N^u$ if the same limits exist for a (lower semicontinuous) rate function $I^u$.

In most physical applications, the rate function is continuous and the upper and lower bounds above turn out to be the same for `normal' sets (typically intervals or compact sets). In this case, we can express the LDP for a macrostate $M_N$ taking value in $\reals$ simply as
\be
\lim_{N\ra\infty}-\frac{1}{N}\ln P_{N,\beta}\{M_N\in [m-r,m+r]\}=I_\beta(m)
\label{eqldp3}
\ee
with the limit $r\ra 0$ implicit as before. For $M_N\in\reals^d$, $[m-r,m+r]$ is replaced by a ball $B_r(m)=\{m'\in \cM: \|m-m'\|\leq r \}$ centered at $m$ to obtain the same result with $N\ra\infty$ followed by $r\ra 0$. In both cases, it is convenient to summarize the limit defining the LDP using the \emph{logarithmic equivalence} notation
\be
P_{N,\beta}\{M_N\in dm\}\asymp e^{-NI_\beta(m)}\, dm
\label{eqldp4}
\ee
where $dm$ denotes an infinitesimal interval or ball centered at $m$ \cite{ellis1985,ellis1995,ellis1999}. This way, we emphasize the two fundamental properties of the LDP, namely the exponential decay of probabilities with $N$, except at points where the rate function vanishes, and the fact that this decay is in general only \emph{approximately} exponential in $N$, that is, exponential in $N$ up to first order in the exponent.\footnote{In information theory, the sign $\doteq$ is sometimes used instead of $\asymp$ \cite{cover1991}.} A similar notation obviously holds for the microcanonical LDP.

The simplified LDPs in (\ref{eqldp3}) and (\ref{eqldp4}) are convenient for expressing the results of this paper, but are not used for proving these results. All the LDPs stated in the following with $\asymp$ are shorthand for the full definition of the LDP given above, with the upper and lower bounds, due to Varadhan \cite{varadhan1966}. Moreover, though most results about macrostates are stated for $\cM=\reals^{d}$, they can be strengthen, as mentioned, to a Polish space $\cM$. For more details about LDPs defined in the context of statistical mechanics, the lower semicontinuity of rate functions, and the $\asymp$ notation, see \cite{ellis1985,ellis1995,ellis1999,touchette2009}. 

\subsection{Equilibrium macrostates}
\label{secdefeqmac}

The LDP of $M_N$ with respect to $P_N^u$ and $P_{N,\beta}$ imply, as mentioned before, that these measures concentrate exponentially with $N$ on certain points of $\cM$ corresponding physically to the \emph{typical} or \emph{equilibrium} values of $M_N$ obtained in the thermodynamic limit. Mathematically, these points must correspond to minima and zeros of $I^u$ and $I_\beta$, since rate functions are always non-negative \cite{ellis1985}. This justifies defining the set $\cE^u$ of equilibrium values of the macrostate $M_N$ in the microcanonical ensemble at mean energy $u$ as
\be
\cE^u=\{m\in\cM:I^u(m)=0\}
\ee
and the set of equilibrium values of $M_N$ in the canonical ensemble at inverse temperature $\beta$ as
\be
\cE_\beta=\{m\in\cM:I_\beta(m)=0\}.
\ee
The former definition formalizes Einstein's observation that microcanonical equilibrium states maximize the macrostate entropy, identified here as $-I^u(m)$, whereas the latter formalizes Landau's later observation that canonical equilibrium states minimize the canonical macrostate free energy, corresponding here to the rate function $I_\beta(m)$. More information about these definitions and identifications can be found in \cite{ellis1985,ellis1995,ellis1999}, \cite{lewis1994a,lewis1994,lewis1995}, and Secs.~5.3 and 5.4 of \cite{touchette2009}. In \cite{lewis1994a,lewis1994,lewis1995}, the equilibrium macrostate sets are called \emph{concentration sets}.

As observed by Lanford \cite{lanford1973} (see also \cite{ellis2000,lewis1994a,lewis1994,lewis1995}), the interpretation of the elements of $\cE^u$ or $\cE_\beta$ as equilibrium states is rigorously justified when these sets contain one element. In this common case, it is relatively easy to show that the microcanonical or canonical measure of $M_N$ is exponentially concentrated on a single value, so that $M_N$ converges in probability to this typical value in the limit $N\ra\infty$. A proof of this result, which establishes a Law of Large Numbers for $M_N$, can be found for example in Theorems 2.5 and 3.6 of \cite{ellis2000}.

There is a problem, however, when $\cE^u$ or $\cE_\beta$ contains more than one elements for a given value of their parameters, which arises typically when there is phase coexistence in phase transitions. In this case, there are two possibilities: i) all the elements of $\cE^u$ or $\cE_\beta$ are concentration points of $M_N$ and so correspond to `real' equilibrium states; or ii) some of these elements correspond to points where the probability of $M_N$ decays sub-exponentially with $N$. Here, we consider the first case; the second involves corrections to the LDP of $P_N^u$ or $P_{N,\beta}$ far beyond the scope of this paper. For a discussion of these corrections in the context of the 2D Ising model, see Examples 5.4 and 5.6 of \cite{touchette2009} and the references cited therein.

\subsection{Thermodynamic potentials}

Before we proceed to discuss the equivalence of the microcanonical and canonical ensembles, we need to define two additional functions, corresponding to the thermodynamic potentials of each ensemble. The first is the \emph{canonical free energy} or \emph{specific free energy} defined as
\be
\varphi(\beta)=\lim_{N\ra\infty}-\frac{1}{N}\ln Z_{N}(\beta)
\label{eqcanfr1}
\ee
with $\beta\in\reals$. This function is also sometimes called the \emph{pressure} \cite{ruelle1969} following its interpretation in the grand-canonical ensemble. Its domain is denoted by 
\be
\dom \varphi=\{\beta\in\reals: \varphi(\beta)>-\infty\}.
\ee
In large deviation theory, $\varphi(\beta)$ is up to a sign the so-called \textit{scaled cumulant generating function} \cite{touchette2009} of the mean energy $h_N$ with respect to the prior $P_N$; see (\ref{eqzn1}). We define this function in the mathematical rather than physics way without a $1/\beta$ pre-factor in order for $\varphi(\beta)$ to be everywhere concave. To be more precise, $\varphi(\beta)$ is by definition a finite, concave and upper semicontinuous function \cite{ellis2000}; by concavity, it is also continuous in the interior of its domain; see \cite{rockafellar1970}. 

In the microcanonical ensemble, the thermodynamic potential to consider is the \emph{microcanonical entropy} or \emph{specific entropy}, defined as
\be
s(u)=\lim_{r\ra 0}\lim_{N\ra\infty}\frac{1}{N}\ln P_N\{h_N\in [u-r,u+r]\}=\lim_{N\ra\infty} \frac{1}{N}\ln P_{N}\{h_{N}\in du\},
\label{eqsu1}
\ee
$P_N$ being again the prior measure. The domain of this function is 
\be
\dom s=\{u\in\reals: s(u)>-\infty\}
\ee
and is assumed to coincide with the range of $h_N$, so that $s$ is defined for all possible values of $h_{N}$. It is clear from the large deviation point of view that the definition of $s(u)$ is equivalent to an LDP for the mean energy $h_N$ with respect to $P_N$, which we write without a minus sign to comply with the physics notation. With the asymptotic notation, we thus express this LDP as
\be
P_N\{h_N\in du\}\asymp e^{Ns(u)}du,
\label{eqsu2}
\ee
where $du=[u-r,u+r]$ with $r\ra 0$ as before.

The entropy function $s(u)$ is upper semicontinuous since it is a rate function \cite{ellis1985}, but is not necessarily concave, as often assumed. Following the introduction, it is the system studied and the form of its Hamiltonian $H_N$ that determines whether or not the entropy is concave. For short-range systems, $s(u)$ is concave, but for long-range systems (see \cite{campa2009} for examples), it can be nonconcave. This is the starting point of nonequivalent ensembles.

\section{Thermodynamic equivalence}
\label{secthermo}

We begin our study of the equivalence problem with the thermodynamic level. As mentioned in the introduction, the problem at this level is to determine whether there is a correspondence between the thermodynamic properties of an $N$-particle system obtained in the microcanonical and canonical ensembles via the entropy $s(u)$ and free energy $\varphi(\beta)$, respectively. It is known from thermodynamics that these functions are related by a Legendre transform, so the mathematical question that we need to answer is: \emph{What are the mathematical conditions guaranteeing that the Legendre transform between $s(u)$ and $\varphi(\beta)$ is involutive, that is, self-inverse?} 

These conditions have been studied in many works and relate to the concavity of $s(u)$; see \cite{touchette2004b} for a review. We repeat them in this section to make the presentation self contained and introduce some definitions and concepts of convex analysis that will be used in the next sections. For applications to physical systems having nonconcave entropies, see \cite{campa2009}.

\subsection{Equivalence results}

To discuss the thermodynamic equivalence of the canonical and microcanonical ensembles, we obviously need $\varphi$ and $s$ to exist:
\begin{assumption}[Existence of thermodynamic potentials]\
\begin{enumerate}
\item[\emph{(A1)}] The limit defining $\varphi(\beta)$ exists and yields a function different than $0$ or $\infty$ everywhere;

\item[\emph{(A2)}] $h_N$ satisfies the LDP with respect to the prior measure $P_N$ with entropy function $s(u)$.
\end{enumerate}
\end{assumption}

The assumptions are not independent, for if $s$ exists, then $\varphi$ also exists and is given by the \emph{Legendre-Fenchel transform} (or conjugate)  of $s$:
\be
\varphi(\beta)=\inf_{u\in\reals}\{\beta u-s(u)\}.
\label{eqlf1}
\ee 
This result implies our first result about equivalence, namely: the canonical thermodynamic behavior of a system,  as encoded in $\varphi(\beta)$ as a function of the inverse temperature of a heat bath, can always be determined from the microcanonical ensemble knowing $s(u)$. The rigorous proof of this transform follows using Varadhan's generalization of the Laplace integral approximation (also known as the Laplace principle) reproduced in Appendix~\ref{appvar}; see also Sec.~II.7 of \cite{ellis1985} and Sec.~4 of \cite{ellis2000}. Following the theory of convex functions, we denote this Legendre transform by $\varphi=s^*$ \cite{rockafellar1970}.

What is interesting for the equivalence problem is that the inverse transform does not always hold. To see this, define the Legendre-Fenchel transform of $\varphi$, which corresponds to the double Legendre-Fenchel transform $(s^*)^*=s^{**}$ of $s$:
\be
s^{**}(u)=\inf_{\beta\in\reals}\{\beta u-\varphi(\beta)\}=\inf_{\beta\in\reals}\{\beta u-s^*(\beta)\}.
\ee
This is a concave and upper semicontinuous function such that $s^{**}(u)\geq s(u)$ for all $u\in\dom s$, corresponding geometrically to the \emph{concave envelope} or \emph{concave hull} of $s(u)$ \cite{rockafellar1970}. As a result, if $s$ is concave, then $s=s^{**}$ and the Legendre-Fenchel transform is dual: $\varphi=s^*$ and $s=\varphi^*$. In this case, we say that we have \emph{thermodynamic equivalence}, since $\varphi$ and $s$ can be transformed into one another. However, if $s$ is not concave, that is, if $s\neq s^{**}$, then there are some parts of $s$ that do not correspond to the Legendre-Fenchel transform of $\varphi$ and thus cannot be obtained from the canonical ensemble. In this case, we say that we have \emph{thermodynamic nonequivalence of ensembles} \cite{ellis2000}. Physically, this means that a system having a nonconcave entropy must have thermodynamic properties as a function of the mean energy that cannot be accounted for within the canonical ensemble as a function of temperature (for otherwise $s$ and $\varphi$ would be related by Legendre-Fenchel transform).

This definition of thermodynamic equivalence is \emph{global}, since it is based on the whole of $s$ and $\varphi$. A \emph{local} definition can also be given by comparing $s(u)$ and $s^{**}(u)$ for specific values of the mean energy $u$. In this case, it is convenient to define $s(u)$ as being \emph{concave} at $u\in\dom s$ if $s(u)=s^{**}(u)$ and \emph{nonconcave} otherwise. 

\begin{definition}[Thermodynamic equivalence]\
\begin{enumerate}
\item[\emph{(a)}] If $s(u)$ is concave at $u$, then the microcanonical ensemble at mean energy $u$ is said to be \emph{thermodynamically equivalent} with the canonical ensemble;
\item[\emph{(b)}] If $s(u)$ is nonconcave at $u$, then the microcanonical ensemble at mean energy $u$ is \emph{thermodynamically nonequivalent} with the canonical ensemble.
\end{enumerate}
\end{definition}

\begin{figure*}
\includegraphics{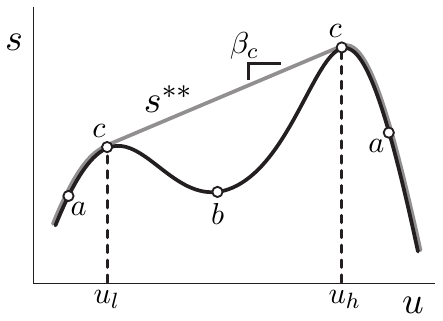}%
\hspace*{1in}%
\includegraphics{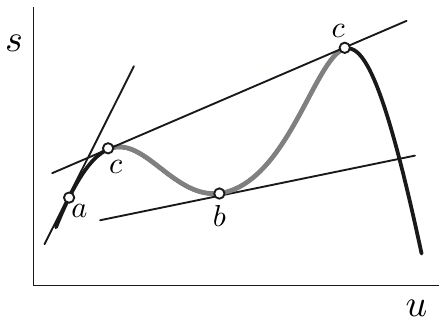}
\caption{Left: Nonconcave entropy $s(u)$ and its concave envelope $s^{**}(u)$. Points $a$: Strictly concave point of $s(u)$ admitting a supporting line that does not touch other points of $s$. Point $b$: Nonconcave point of $s(u)$ with no supporting line. Points $c$: Non-strictly concave points of $s(u)$ with a supporting line touching more than one point of $s$. See Appendix~\ref{appconc} for the complete definitions.}
\label{fignoneq1}
\end{figure*}

Note that, although the term or relation `equivalent' is symmetric (as in `equal'), there is a directionality in the interpretation of equivalence in that, as noted before, the whole of $\varphi$ can always be obtained by Legendre-Fenchel transform from $s$, but $s$ cannot always be obtained from $\varphi$. We will see in the next section that this property leads to a similar `directionality' in the equivalence of the microcanonical and canonical ensembles at the macrostate level.

The next result gives a more geometric characterization of thermodynamic equivalence based on subdifferentials and supporting lines. These concepts, which are important for the next sections, are defined in Appendix~\ref{appconc} and illustrated in Fig.~\ref{fignoneq1}. For the purpose of this paper, the main property of concave points of $s$ to note is that they admit supporting lines except possibly at boundary points of $\dom s$; see Appendix~\ref{appconc} for more details. 

\begin{proposition}
Except possibly at boundary points of $s$, we have the following:
\begin{enumerate}
\item[\emph{(a)}] If $s$ admits a supporting line at $u$, then the microcanonical ensemble at $u$ is thermodynamically equivalent with the canonical ensemble for all $\beta\in\partial s(u)$;
\item[\emph{(b)}] If $s$ does not admit a supporting line at $u$, then the microcanonical ensemble at $u$ is thermodynamically nonequivalent with the canonical ensemble for all $\beta\in\reals$.
\end{enumerate}
\end{proposition}

Part (a) is of course well known in statistical physics in the form of the Legendre transform 
\be
s(u)=\beta u-\varphi(\beta),\quad \beta=s'(u)
\ee
or more simply $F=E-TS$. The more general characterization of $\beta$ in terms of the subdifferential $\partial s(u)$ arises because $s(u)$ is not necessarily differentiable and follows from the fact that the Legendre-Fenchel transform is equal to 
\be
s(u)=\beta u-\varphi(\beta)
\ee
for all $\beta\in\partial s(u)$ and all $u\in\dom s$ such that $\partial s(u)\neq\emptyset$; see Theorem 23.5 of \cite{rockafellar1970} or Theorem A.4 of \cite{costeniuc2005}. For references on part (b), see \cite{touchette2004b}.

\subsection{Energy-temperature relation}

The thermodynamic equivalence of the canonical and microcanonical ensembles can be understood physically by comparing the mean energies of the two ensembles in the thermodynamic limit. Since $h_N$ is a random variable in the canonical ensemble, this ensemble and the microcanonical ensemble, with its fixed mean energy $u$, cannot be equivalent for $N<\infty$. However, the physical expectation is that the canonical measure of $h_N$ concentrates in the thermodynamic limit around some equilibrium value $u_\beta$ of the mean energy, which can be related for a given $\beta$ to the mean energy $u$ of the microcanonical ensemble.

This reasoning, due to Gibbs \cite{gibbs1902}, can be found in almost all textbooks of statistical mechanics as the basis for stating that the canonical and microcanonical ensembles must become equivalent in the thermodynamic limit. To establish this reasoning as a rigorous result, we must determine the set of equilibrium values of $h_N$ in the canonical ensemble and see if they can indeed be related to the mean energy of the microcanonical ensemble \cite{lewis1994,touchette2004b}. This is done in the next two results, which relate this problem to the concave points of $s$ and, consequently, to the involutiveness of the Legendre-Fenchel transform between $s$ and $\varphi$.

\begin{proposition} 
\label{rescanldp1}
Under Assumptions A1-A2, $h_N$ satisfies the LDP in the canonical ensemble with respect to $P_{N,\beta}$ with rate function
\be
J_\beta(u)=\beta u-s(u)-\varphi(\beta).
\label{eqjb1}
\ee
\end{proposition}

This result is proved heuristically in \cite{touchette2003} and Example 5.5 of \cite{touchette2009}. A rigorous proof is given next based on Varadhan's Theorem (see Appendix~\ref{appvar}) and applies for a general function $h_{N}$ defined in a Polish space. For a similar result obtained for observables other than the mean energy, see Theorem 4.1 of \cite{lewis1995}.

\begin{proof}
This result follows directly from Theorem~\ref{thmtvar1} of Appendix~\ref{appvar} with the following substitutions: $-s(u)$ takes the role of $I(x)$, $P_{N,\beta}$ takes the role of $P_{n,F}$, and $F(h_N)=-\beta h_N$. Although the latter function is not bounded, we have $\varphi(\beta)<\infty$ by Assumption A1. Therefore, the result of this theorem applies and yields that $h_N$ satisfies the LDP with respect to $P_{N,\beta}$ with rate function
\be
\beta u -s(u)-\inf_{u}\{\beta u -s(u)\}=\beta u-s(u)-\varphi(\beta).
\ee
\end{proof}

Following the logarithmic notation introduced earlier, we express the LDP of Proposition~\ref{rescanldp1} as
\be
P_{N,\beta}\{h_N\in du\}\asymp e^{-NJ_\beta(u)}\, du
\ee 
to emphasize the exponential concentration of the canonical measure of $h_N$ and the fact that the canonical equilibrium values of $h_N$ must correspond to the zeros of $J_\beta$. Let $\cU_\beta$ denote the set of these equilibrium values obtained for a given inverse temperature $\beta$, that is,
\be
\cU_\beta=\{u\in\reals:J_\beta(u)=0\}.
\ee
Note that the minimizers of $J_\beta$ are necessarily in $\dom s$. From the explicit form of this rate function, we obtain the following relation between the elements of $\cU_\beta$ and the entropy:

\begin{proposition}
\label{lemu1}
Assume A1-A2. Then $u\in\cU_\beta$ if and only if $\beta\in\partial s(u)$. 
\end{proposition}

\begin{proof}
We first prove the necessary part of this result. Let $u\in\dom s$ and assume that $\beta\in\partial s(u)$. Then, by the definition of subdifferentials (see Appendix~\ref{appconc}), we have
\be
s(v)\leq s(u)+\beta(v-u)
\ee
for all $v\in\reals$. Equivalently, 
\be
\beta u-s(u)\leq \beta v-s(v),
\ee
which implies from (\ref{eqjb1}) that $u$ is a global minimum of $J_\beta(u)$. Therefore, $u\in\cU_\beta$.

For the sufficiency part, choose $u\in\cU_\beta$, so that $J_\beta(u)=0$. Since $J_\beta(v)\geq 0$ for all $v\in\reals$, we must have
\be
\beta u-s(u)\leq \beta v-s(v)
\ee
for all $v$, which implies $\beta\in\partial s(u)$ by definition of subdifferentials.
\end{proof}

The physical meaning of Proposition~\ref{lemu1} is clear by considering two cases~\cite{touchette2004b}:
\begin{enumerate}
\item If $s$ is strictly concave and differentiable at $u$ (see Appendix~\ref{appconc} for the definition of strictly concave), then $\cU_\beta=\{u\}$ for $\beta=s'(u)$, which means that $u$ is the unique equilibrium mean energy in the canonical ensemble at $\beta$. In this case, Gibbs's reasoning is valid: there is equivalence of ensembles in the expected physical way with the temperature-entropy relation $\beta=s'(u)$ arising from the (involutive) Legendre transform between $\varphi$ and $s$.

\item If $s$ is nonconcave at $u$, then $u\notin \cU_\beta$ for any $\beta\in\reals$. In this case, illustrated in Fig.~\ref{fignoneq2}, Gibbs's reasoning does not work: there is nonequivalence of ensembles because there is no inverse temperature in the canonical ensemble that yields $u$ as the equilibrium mean energy in the canonical ensemble. 
\end{enumerate}

The second case also implies, as illustrated in Fig.~\ref{fignoneq2}, that the nonconcave region of $s$ is `skipped over' by the canonical mean energy $u_{\beta}$ and, therefore, that there must be a discontinuous phase transition in the canonical ensemble. For more precise results on this relation between nonequivalent ensemble and canonical first-order phase transitions, see \cite{touchette2005a}, Theorem 4.10 of \cite{ellis2000} or Sec.~5.5 of \cite{touchette2009}. When $s(u)$ is twice differentiable, a further relation can be proved between nonequivalent ensembles and energies where the microcanonical heat capacity, defined as
\be
c(u)=\frac{du}{dT(u)}=\frac{du}{d(s'(u)^{-1})}=-\frac{s'(u)^{2}}{s''(u)},
\label{eqheatcap1}
\ee
becomes negative; see Sec.~3.4 of \cite{touchette2003} and Theorem 2.2 of \cite{touchette2005a}. Finally, for examples of long-range systems having nonconcave entropies and nonequivalent ensembles, see \cite{campa2009}.

\begin{figure*}
\includegraphics{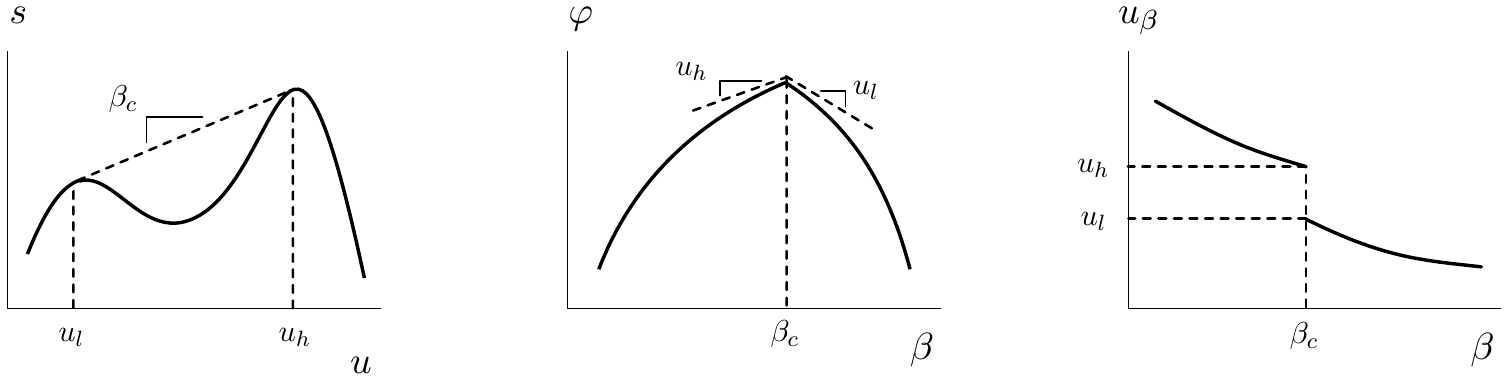}
\caption{Left: Nonconcave entropy $s(u)$. Center: Associated free energy $\varphi(\beta)$ having a nondifferentiable point at $\beta_c$. Right: Equilibrium mean energy $u_\beta$ in the canonical ensemble as a function of $\beta$: as $\beta$ is varied continuously, $u_\beta$ jumps over the nonconcave interval $(u_l,u_h)$, giving rise to a first-order phase transition in the canonical ensemble with a specific latent heat $\Delta u=u_h-u_l$.}
\label{fignoneq2}
\end{figure*}

\section{Macrostate equivalence}
\label{secmacro}

We now discuss the equivalence of ensembles at the deeper level of equilibrium macrostates. Our previous discussion of the equilibrium values of $h_N$ in the canonical ensemble is already a form of macrostate equivalence relating the elements of $\cU_\beta$ to the control parameter $u$ of the microcanonical ensemble. In this section, we study the macrostate level more generally by comparing the two equilibrium sets $\cE_\beta$ and $\cE^u$ for a general macrostate $M_N$.

This level of equivalence was studied for a class of general macrostates by Ellis, Haven and Turkington \cite{ellis2000} following previous results by Eyink and Spohn \cite{eyink1993}. Other authors \cite{deuschel1991,stroock1991,stroock1991b,roelly1993}, including Georgii \cite{georgii1993,georgii1994,georgii1995}, have derived important results about the equivalence of ensembles for an abstract, infinite-dimensional macrostate known in large deviation theory as the \emph{empirical process} or \emph{level 3 of large deviations} involving the relative entropy; see \cite{ellis1985,ellis1995,ellis1999} for details. Except for \cite{ellis2000}, however, all these works assume that entropy is concave, as mentioned in the introduction. Moreover, although other macrostates can be obtained in principle by contraction of the empirical process, it is very difficult in practice to use this contraction to derive LDPs for simple physical macrostates such as the magnetization and the energy.

In this section, we propose a simple approach to macrostate equivalence, which applies to general systems and macrostates and which shows that equivalence holds at this level under the same conditions as thermodynamic equivalence. The presentation follows the results announced in \cite{touchette2011b}, which generalize those of Ellis, Haven and Turkington \cite{ellis2000} to any classical $N$-body systems and macrostates $M_N$ under the following assumptions:
\begin{assumption}[Existence of equilibrium macrostates]\
\begin{enumerate}
\item[\emph{(A3)}] $M_N$ satisfies the LDP with respect to the canonical measure $P_{N,\beta}$ with rate function $I_\beta$ for all $\beta\in\dom \varphi$;
\item[\emph{(A4)}] $M_N$ satisfies the LDP with respect to the microcanonical measure $P_N^u$ with rate function $I^u$ for all $u\in\dom s$.
\end{enumerate}
\end{assumption}
These assumptions, which also require A1-A2, are obviously weak: they are there only to make sure that $\cE_\beta$ and $\cE^u$ exist, are non-empty (by lower semicontinuity of rate functions) and so can be compared. The problem of identifying classes of Hamiltonians for which these assumptions are verified is a  more difficult problem, which we do not address here. For long-range systems, this problem remains completely open.

\subsection{Canonical ensemble as a mixture of microcanonical ensembles}

The microcanonical and canonical ensembles are based on two obviously different probability measures on $\Lambda_N$: the former assigns a non-zero measure to microstates of a given energy, whereas the latter assigns a non-zero measure to all $\om\in\Lambda_N$. However, the two are fundamentally related, in that the canonical ensemble can be expressed as a `probabilistic mixture' of microcanonical ensembles. This is the key insight needed to obtain general results about macrostate equivalence.

To explain what we mean by a mixture of ensembles, consider the canonical probability measure $P_{N,\beta}(d\om)$ defined in (\ref{eqcan1}). Since this measure depends only the product $\beta H_N(\om)$, it is clear that all microstates having the same energy have the same probability. As a result, the conditional probability measure $P_{N,\beta}\{d\om|h_N \in u\}$ obtained by conditioning $P_{N,\beta}(d\om)$ on the set of microstates such that $h_N(\om)\in u$ must be `uniform' over that constrained set of microstates. This is obvious from the definition of this conditional measure:
\begin{eqnarray}
P_{N,\beta}\{ d\om|h_N\in du\} &=& \frac{P_{N,\beta}\{d\om,h_N\in du\}}{P_{N,\beta}\{h_N\in du\}}\nonumber\\
&=&\frac{e^{-\beta Nu}}{Z_N(\beta)}\frac{\id_{du}(h_N(\om))}{P_{N,\beta}\{h_N\in du\}}\, dP_N(\om),
\label{eqcondcan1}
\end{eqnarray}
where for the second line we have used the fact that $h_{N}=u$ in the limit $r\ra 0$. Thus, we see that $P_{N,\beta}\{ d\om|h_N\in du\} $ is proportional to $P^u_N(d\om)$, as defined in (\ref{eqmicro1}), so that 
\be
P_{N,\beta}\{ d\om|h_N\in du\}=P_N^u\{d\om\},
\label{eqcanmicro1}
\ee
for all $\om\in\Lambda_{N}$, since both measures are normalized to 1. Incidentally, normalizing (\ref{eqcondcan1}) yields
\be
P_{N,\beta}\{h_N\in du\}=\frac{e^{-\beta Nu }}{Z_N(\beta)} P_N\{h_N\in du\}.
\label{eqmicronorm1}
\ee
Taking the large deviation limit then gives the result of Proposition~\ref{rescanldp1}.

With the basic equality (\ref{eqcanmicro1}), it now follows from
\be
P_{N,\beta}(d\om)=\int_\reals P_{N,\beta}\{d\om | h_N=u\} \, P_{N,\beta}\{h_N\in du\},
\ee
that
\be
P_{N,\beta}(d\om)=\int_\reals P_{N}^u (d\om) \, P_{N,\beta}\{h_N\in du\}.
\label{eqmix1}
\ee
Applying this result to an arbitrary macrostate $M_N$ then yields
\be
P_{N,\beta}\{ M_N\in A\}=\int_{\reals} P_N^u\{M_N\in A\}\, P_{N,\beta}\{h_N\in du\}
\label{eqmicrocan2}
\ee
for any measurable set $A$. Hence, we see that the canonical measure on both $\Lambda_N$ and $\cM$ is a superposition of microcanonical measures weighted by the canonical mean energy distribution $P_{N,\beta}(du)=P_{N,\beta}\{h_N\in du\}$. It is this superposition, which is exact for any $N<\infty$, that we refer to as a probabilistic mixture of microcanonical ensembles. 

In what follows, we use this result to relate the equilibrium states of the microcanonical ensembles to those of the canonical ensemble. Already, it should be clear that (\ref{eqmicrocan2}) implies a link between the different LDPs of these ensembles: we know from Proposition~\ref{rescanldp1} that $P_{N,\beta}(du)$ satisfies the LDP with rate function $J_\beta(u)$, whereas $P_{N,\beta}\{ M_N\in A\}$ and $P_{N}^u\{ M_N\in A\}$ both satisfy the LDP by assumption. Exploiting the exponential form of these LDPs in the mixture integral (\ref{eqmicrocan2}), we then obtain the following result:

\begin{proposition}
\label{theoremmain}
Under Assumptions A1-A4, 
\be
I_\beta(m)=\inf_{u\in\reals}\{I^u(m)+ J_\beta(u)\}
\label{eqrep1}
\ee
for any $m\in\cM$ and $\beta\in\dom\varphi$. The minimizers in this formula are necessarily in $\dom s$.
\end{proposition}

This result was first announced in \cite{touchette2011b} and is proved there heuristically following the argument based on the exponential form of (\ref{eqmicrocan2}) just mentioned. We give next a rigorous form of this argument based on the contraction principle and Varadhan's version of the Laplace principle.

\begin{proof}
By assumption, $M_N$ satisfies the LDP in the microcanonical ensemble with rate function $I^u$, while $h_N$ satisfies the LDP in the canonical ensemble with rate function $J_\beta$. The product $P_N^u \{M_N\in A\} P_{N,\beta}(du)$ in (\ref{eqmicrocan2}) therefore satisfies a joint LDP for $(M_{N},h_{N})$ with rate function $I^u+J_\beta$ by definition of the LDP. The marginalization of $h_{N}$ in this joint LDP, corresponding to the integral in (\ref{eqmicrocan2}), then yields to (\ref{eqrep1}) via the contraction principle \cite{dembo1998}.

Alternatively, we can apply the result (\ref{eqvar1}) with $F=0$ to the integral of (\ref{eqmicrocan2}) with the joint LDP to obtain the Laplace approximation
\be
\liminf_{N\ra\infty}-\frac{1}{N}\ln\int_\reals P_N^u \{M_N\in C\} P_{N,\beta}(du)\geq \inf_{m\in C}\inf_{u\in\reals}\{ I^u(m)+J_\beta(u)\}
\ee
for $C$ closed and
\be
\limsup_{N\ra\infty}-\frac{1}{N}\ln\int_\reals P_N^u \{M_N\in O\} P_{N,\beta}(du)\leq \inf_{m\in O}\inf_{u\in\reals}\{ I^u(m)+J_\beta(u)\}
\ee
for $O$ open. The result (\ref{eqrep1}) then follows because rate functions are unique \cite{ellis1985}. 

Both arguments work not only for $M_{N}\in\reals^{d}$ and $h_{N}\in\reals$, but also for $M_{N}$ and $h_{N}$ taking values in Polish spaces. Moreover, the infimum over $\reals$ can be restricted to $\dom s$, since $I^u$ and $J_\beta$ are by assumption infinite outside $\dom s$.
\end{proof}

Proposition~\ref{theoremmain} relates the fluctuations of the canonical and microcanonical ensembles for general macrostates. This result is interesting physically as it shows that these fluctuations depend on the system and macrostate considered, and so cannot be expected to be the same in general.\footnote{Consider the obvious example of the mean energy $h_N$, which does not fluctuate in the microcanonical ensemble but does in the canonical ensemble.} For this reason, one cannot speak of the equivalence of ensembles in terms of macrostate fluctuations \cite{yukalov2005}, only in terms of their equilibrium macrostates.

\subsection{Equivalence results}

The general result (\ref{eqrep1}) relates not only the fluctuations of the microcanonical and canonical ensembles, but also their equilibrium states. Since rate functions are nonnegative, $I_\beta(m)$ vanishes if and only if both $I^u(m)$ and $J_\beta(u)$ vanish in (\ref{eqrep1}). This implies that the equilibrium values of $M_N$ in the canonical ensemble must correspond to the equilibrium values of $M_N$ in the microcanonical ensemble for all mean energies realized at equilibrium in the canonical ensemble. This is stated in the next result.
\begin{proposition}
\label{theoremeqens1}
Under Assumptions A1-A4:
\be
\cE_\beta=\bigcup_{u\in\cU_\beta}  \cE^u.
\label{eqdec1}
\ee
\end{proposition}

\begin{proof}
Take $m\in\cE_\beta$. Then $I_\beta(m)=0$ by definition of $\cE_\beta$, so that, by Proposition~\ref{theoremmain},
\be
0=\inf_u \{ I^u(m_\beta)+J_\beta(u)\}.
\ee
Since rate function are nonnegative, this implies that there exists $u\in\dom s$ such that $I^u(m)=0$, implying $m\in \cE^u$, and $J_\beta(u)=0$, so that $u\in \cU_\beta$. As this is true for all elements of $\cE_\beta$, we obtain
\be
\cE_\beta \subseteq \bigcup_{u\in \cU_\beta} \cE^u =\cE^{\cU_\beta}.
\label{eqincl1}
\ee

We now prove the reverse inclusion. Consider $u\in\cE_\beta$ for which $J_\beta(u)=0$ and $m\in\cE^u$ for which $I^u(m)=0$. Then the result of Proposition~\ref{theoremmain} gives $I_\beta(m)=0$, so that $m\in\cE_\beta$. As this is true for all $m\in\cE^u$ with $u\in\cU_\beta$, we obtain
\be
\bigcup_{u\in \cU_\beta} \cE^u \subseteq \cE_\beta.
\ee
Therefore, the two sides are equal.
\end{proof}

The covering result (\ref{eqdec1}) shows that the canonical equilibrium macrostates are always realized in the microcanonical ensemble for one or more values of $h_N$. To determine when $\cE^u$ coincides with $\cE_\beta$ for some $\beta$, we next use Proposition~\ref{lemu1} to determine whether $\cU_\beta$ has one element, many elements, or is empty. This leads us to the following result about macrostate equivalence, which is the main result of this section.

\begin{theorem}[Macrostate equivalence]
\label{theoremellis2}
Assume A1-A4. Then
\begin{enumerate}
\item[\emph{(a)}] \emph{Strict equivalence:} If $s$ is strictly concave at $u$, then $\cE^u=\cE_\beta$ for some $\beta\in\reals$;
\item[\emph{(b)}] \emph{Nonequivalence:} If $s$ is nonconcave at $u$, then $\cE^u\neq \cE_\beta$ for all $\beta\in\reals$;
\item[\emph{(c)}] \emph{Partial equivalence:} If $s$ is concave but not strictly concave at $u$, then $\cE^u\subseteq \cE_\beta$.
\end{enumerate}
\end{theorem}

\begin{proof}
Case (a): This follows from the result stated after Proposition~\ref{lemu1} that, if $s$ is strictly concave at $u$, then $\cU_\beta$ is the singleton set $\{u\}$ for $\beta\in\partial s(u)$. From the covering result (\ref{eqdec1}), we then obtain $\cE_\beta=\cE^u$ for all $\beta\in\partial s(u)$. 

Case (b):
The assumption that $s$ is nonconcave at $u$ implies also from Proposition~\ref{lemu1} that $u\notin\cU_\beta$ for all $\beta\in\reals$. Let $m^u\in\cE^u$ and assume that $m^u\in\cE_\beta$ for some $\beta\in\reals$. Then using (\ref{eqdec1}), or equivalently the relation (\ref{eqrep1}), we must have $u\in\cU_\beta$, which contradicts the result that $u\notin\cU_\beta$. Since this contradiction is reached for any $m^u\in\cE^u$ and any $\beta\in\reals$, we conclude that $\cE^u\cap\cE_\beta=\emptyset$ for all $\beta\in\reals$, a result which we write as $\cE^u\neq \cE_\beta$ for all $\beta$.

Case (c):
If $s$ is concave at $u$ but non-strictly concave, then $u\in\cU_\beta$ for $\beta\in\partial s(u)$, but $\cU_\beta$ is no longer a singleton: by definition of non-strict concave points, there must exist at least one $v\neq u$ for which $\beta\in\partial s(v)$ and so for which $v\in\cU_\beta$. In this case, the covering result (\ref{eqdec1}) involves at least two sets, which implies that $\cE^u\subseteq \cE_\beta$ in general. If, as in most systems, $\cE^v\neq\cE^u$, then this inclusion is strengthened to $\cE^u\subsetneq \cE_\beta$, that is, $\cE^u$ is a proper subset of $\cE_\beta$.
\end{proof}

Cases (a) and (b) have clear interpretations in terms of Gibbs's reasoning \cite{touchette2004b}. 

In case (a), the microcanonical and canonical ensemble are equivalent at the macrostate level because the mean energy of the latter ensemble is concentrated on a single value corresponding to $u$ for $\beta\in\partial s(u)$. For $s(u)$ differentiable, $\beta$ and $u$ are then related by the standard thermodynamic relation $\beta=s'(u)$, as already mentioned after Proposition~\ref{lemu1}. 

In case (b), we have nonequivalence because $u$ is never realized in the canonical ensemble as an equilibrium mean energy, so that the set $\cE^u$, which can be realized in the microcanonical ensemble by fixing $h_N=u$, cannot be realized in the canonical ensemble by varying $\beta$ instead. In this case, it can in fact be proved, under additional assumptions, that the elements of $\cE^{u}$ correspond either to unstable or metastable critical points of $I_{\beta}(m)$, depending on the sign of the microcanonical heat capacity $c(u)$ defined in (\ref{eqheatcap1}); see \cite{touchette2005a} for more details.

Case (c) is more subtle: it arises when $\cU_\beta$ has more than one element, and so when the canonical ensemble has many coexisting equilibrium mean energies, giving rise at the macrostate level to many coexisting equilibrium macrostates, called \emph{phases} in statistical mechanics. The next result, which follows from the theorem above, shows that this naturally arises whenever $s(u)$ is nonconcave or has some linear parts.

\begin{corollary} If $s(u)$ is nonconcave or is non-strictly concave, then there exists $\beta_c\in\reals$ such that $\cE_{\beta_c}$ is composed of two or more microcanonical sets $\cE^u$ with $u\in\cU_{\beta_c}$, i.e., $\cE_{\beta_c}=\cE^u\cup\cE^{u'}\cup\cdots$, with $u,u',\ldots\in\cU_{\beta_c}$.
\end{corollary}

Figure~\ref{fignoneq2} illustrates the case, commonly encountered in long-range systems, in which two phases appear at some critical inverse temperature $\beta_c$ due to the nonconcavity of $s(u)$ over some interval $(u_l,u_h)$, leading to $\cU_{\beta_c}=\{u_l,u_h\}$ and $\cE_{\beta_c}=\cE^{u_l}\cup \cE^{u_h}$. The relation between the nonconcave region of $s(u)$ and the nondifferentiability of $\varphi(\beta)$ is illustrated in Fig.~\ref{fignoneq2}. If $s$ has linear or affine parts over $(u_l,u_h)$, then a similar physical interpretation involving a first-order phase transition also applies, but with discrete phases replaced by a continuum of phases \cite{touchette2009}.

\subsection{Comparison with previous results}

Proposition~\ref{theoremeqens1} and Theorem~\ref{theoremellis2} above are generalizations of two results obtained by Ellis, Haven and Turkington \cite{ellis2000}. We have derived these results above by assuming that $\cE_\beta$ and $\cE^u$ exist and by using the idea of probabilistic mixture of microcanonical ensembles to relate these sets. Ellis, Haven, and Turkington use a different approach: they explicitly construct the rate functions $I_\beta$ and $I^u$ and then relate with these $\cE_\beta$ and $\cE^u$. In doing so, they assume the following:
\begin{enumerate}
\item There exists a function $\rh:\cM\ra\reals$ such that
\be
\lim_{N\ra\infty} \big|h_N(\om)-\rh(M_N(\om))\big|=0
\ee
uniformly over all $\om\in\Lambda_N$. In this case, we say that $h_N$ admits an \emph{energy representation function} in terms of $M_N$. 
\item $M_N$ satisfies the LDP with respect to the prior measure $P_N$.
\end{enumerate}

The explicit expressions of $I_\beta$ and $I^u$ obtained under these assumptions can be found in Theorem 2.4 and Theorem 3.2 of \cite{ellis2000}, respectively. In terms of $\rh$, our Theorem~\ref{theoremeqens1} then corresponds to Theorem 4.10 of \cite{ellis2000}, which has the form
\be
\cE_\beta=\bigcup_{u\in\rh (\cE_\beta)}  \cE^u,
\label{eqdec2}
\ee
while our Theorem~\ref{theoremellis2} corresponds essentially to Theorem 4.4 of \cite{ellis2000}. The difference between the two sets of results is that $\cU_\beta$ is replaced by $\rh(\cE_\beta)$.

It can be shown directly using the contraction principle \cite{touchette2009} that if $\rh$ exists, then $\cU_\beta=\rh(\cE_\beta)$  so that the covering results of (\ref{eqdec1}) and (\ref{eqdec2}) are equivalent. What we have shown here is that (\ref{eqdec1}) does not require any of the assumptions above to hold. Consequently, these must be sufficient but not necessary conditions for macrostate equivalence to be related to thermodynamic equivalence. This relation is based only on the existence of thermodynamic functions and equilibrium macrostates, as proved above, and is as such a general result of statistical mechanics.

This generalization is important, as there are many physical systems and macrostates of interest that do not admit an energy representation function. In fact, except for the mean energy $h_{N}$ itself and the empirical process, which can be used to construct $\rh$ for \emph{any} system, pairs $(M_{N},\rh)$ exist only for particular macrostates of non-interacting, mean-field, and some long-range systems; see \cite{eyink1993,ellis2000,campa2009}. An obvious example is the magnetization, which does not admit an energy representation function for short-range systems, such as the 2D Ising model. Similarly, the empirical measure cannot be used to construct $\rh$ for short-range systems and most long-range systems, including gravitating particles \cite{campa2009}. In these cases, equivalence of ensembles can be inferred from our results without $\rh$. For the 2D Ising model, for example, we recover the known result that ensembles are equivalent at the magnetization level \cite{lewis1994}.

\section{Measure equivalence}
\label{secmeas}

The last level of ensemble equivalence that we discuss is concerned with the convergence of $P_{N,\beta}(d\om)$ and $P_{N}^u(d\om)$ at the \emph{micro} rather than \emph{macro}state level. This equivalence is suggested mathematically by the fact that macrostate equivalence (in the strict concavity case) implies 
\be
\lim_{N\ra\infty} E_{P_{N,\beta}}[M_N]=\lim_{N\ra\infty} E_{P_N^u}[M_N]
\ee
for \emph{any} macrostate satisfying Assumptions A1-A4. This mean convergence result is close to the notion of weak convergence and suggests that $P_{N,\beta}$ should converge to $P_N^u$ as $N\ra\infty$ with respect to a `norm' or `metric' that is sensitive to their large deviation properties.

In this section, we consider two such `metrics', the \emph{specific relative entropy} and \emph{specific action}, and show that measure equivalence holds in both cases when $s(u)$ is concave, and thus when there is thermodynamic and macrostate equivalence. Our results for the specific relative entropy are essentially those of Lewis, Pfister and Sullivan \cite{lewis1994a,lewis1994,lewis1995} (see also \cite{stroock1991,stroock1991b,deuschel1991}). New and stronger results are obtained for the specific action, which point interestingly to a general form of the asymptotic equipartition property studied in information theory and the theory of ergodic processes \cite{cover1991}.

\subsection{Relative entropy}

Let $P$ and $Q$ be two probability measures defined on a space $\cX$, and assume that $P$ is absolutely continuous with respect to $Q$ (denoted by $P\ll Q$). The \emph{relative entropy} of $P$ with respect to $Q$ is defined as
\be
D(P||Q)=\int_{\cX} dP(\om) \ln \frac{dP}{dQ}(\om),
\ee
where $dP/dQ$ denotes the \emph{Radon-Nikodym derivative} of $P$ with respect to $Q$. The relative entropy is also called the information gain  \cite{lewis1994a,lewis1994,lewis1995}, the information divergence \cite{csiszar1975,csiszar1984,csiszar1998} or Kullback-Leibler distance \cite{cover1991}. 
Strictly speaking, $D(P||Q)$ is not a distance, since it is not symmetric and does not satisfy the triangle inequality. However, $D(P||Q)\geq 0$ with equality if and only if $P=Q$ almost everywhere \cite{cover1991}. Therefore, it can be interpreted as a generalized metric inducing a well-defined topology on the space of distributions. Moreover, it is known that $D(P||Q)$ is an upper bound on the total variation norm:
\be
d_{TV}(P,Q)=\frac{1}{2}\int_{\cX}\left| dQ-dP\right|\leq \sqrt{D(P|| Q)};
\ee
see, for example, Proposition 10.3 of \cite{lewis1995}.

For the microcanonical and canonical ensembles, we have $P_N^u\ll P_{N,\beta}$ but $P_{N,\beta}\not\ll P_N^u$, since $P_{N}^u$ is a restriction of $P_{N,\beta}$, so that the correct relative entropy to consider is
\be
D(P^u_N||P_{N,\beta})=\int_{\Lambda_N} dP_N^u(\om)\ln \frac{dP_N^u}{dP_{N,\beta}}(\om).
\ee
From this, we define the \emph{specific relative entropy} by the limit
\be
d_\beta^u=\lim_{N\ra\infty} \frac{1}{N}D(P^u_N||P_{N,\beta}).
\label{eqsre1}
\ee
This quantity, when it exists, is also called the relative entropy rate, the specific information gain \cite{lewis1994a,lewis1994,lewis1995} or divergence rate \cite{shields1993}. We use it next to give a first definition of measure equivalence due to Lewis, Pfister and Sullivan \cite{lewis1994a,lewis1994,lewis1995}.

\begin{definition}[Measure equivalence I]
The canonical and microcanonical ensembles are said to be equivalent at the measure level, in the specific relative entropy sense, if $d_\beta^u=0$.
\end{definition}

It is clear from this definition that $d_\beta^u=0$ does not imply $P_N^u(d\om)=P_{N,\beta}(d\om)$ for almost all $\om\in\Lambda_N$; it only implies that $D(P_N^u||P_{N,\beta})$ grows slower than $N$ and, consequently, that the total variation $d_{TV}(P_N^u,P_{N,\beta})$ grows slower than $\sqrt{N}$ as $N\ra\infty$. This, as shown next, is a necessary and sufficient condition for measure equivalence to coincide with thermodynamic and macrostate equivalence.

\begin{theorem}[Measure equivalence I]
\label{thmmeq1}
Under Assumptions A1-A4, $d_\beta^u=0$ if and only if $\beta\in \partial s(u)$. Therefore, except possibly at boundary points of $\dom s$, measure equivalence holds in the specific relative entropy sense if and only if thermodynamic equivalence holds.
\end{theorem}

\begin{proof}
The result follows simply by writing the explicit expression of the Radon-Nikodym derivative of $P_N^u$ with respect to $P_{N,\beta}$:
\be
\frac{dP_N^u}{dP_{N,\beta}}(\om)=\frac{e^{N\beta h_N(\om)} Z_N(\beta)}{P_N\{h_N\in du\}}\, \id_{du}\big( h_N(\om)\big).
\label{eqrndt1}
\ee
Inserting this expression into $D(P_{N}^u||P_{N,\beta})$ and taking the trivial expectation with respect to $P_N^u$ yields
\be
d_\beta^u=\beta u -s(u)-\varphi(\beta) =J_\beta(u),
\ee
where we have also used the limits (\ref{eqcanfr1}) and (\ref{eqsu1}) defining $\varphi(\beta)$ and $s(u)$, respectively. From this result, the statement of the theorem then follows using Proposition~\ref{lemu1} relating the zeros of the canonical rate function $J_\beta(u)$ and the concave points of the microcanonical entropy $s(u)$.
\end{proof}

Part of this theorem can be found in Theorem 5.1 (see also Lemma 5.1) of \cite{lewis1994} and is applied in that work to lattice spin systems, including the mean-field Curie-Weiss model and the 2D Ising model. For related results obtained in the context of 1D and 2D lattice gases, see \cite{stroock1991,stroock1991b,deuschel1991}. Finally, for an application to the nonequilibrium zero-range process, see \cite{grosskinsky2003,grosskinsky2008,chleboun2014}

\subsection{Radon-Nikodym derivative}

We now consider the random variable
\be
R_{N,\beta}^u(\om)=\frac{1}{N}\ln\frac{dP_N^u}{dP_{N,\beta}}(\om),
\label{eqspeca1}
\ee
which depends on the two parameters $\beta$ and $u$. We call this random variable the \emph{specific action} following the definition of a similar quantity for Markov processes \cite{touchette2009}. 

The result proved in Theorem~\ref{thmmeq1} is about the convergence in mean of $R_{N,\beta}^u$ with respect to $P_N^u$. Here we prove a stronger convergence for $R_{N,\beta}^u$ using convergence in probability with respect to both $P_N^u$ and $P_{N,\beta}$, which expresses the concentration of this random variable with respect to both measures. This is basis of our second definition of measure equivalence stated next, which implies the previous one based on the specific relative entropy.

\begin{definition}[Measure equivalence II]
The canonical and microcanonical ensembles are said to be equivalent at the measure level, in the specific action sense, if 
\be
\lim_{N\ra\infty} R_{N,\beta}^u(\om)=0
\ee
almost everywhere with respect to both $P_N^u$ and $P_{N,\beta}$.
\end{definition}

This definition can be expressed differently by saying that the two ensembles are equivalent at the measure level if $P_{N}^u$ and $P_{N,\beta}$ are logarithmically equivalent almost everywhere with respect to these measures, that is, $P_N^u(d\om)\asymp P_{N,\beta}(d\om)$ or, equivalently,
\be
\frac{dP_N^u}{dP_{N,\beta}}(\om)\asymp 1
\ee
almost everywhere with respect to $P_{N}^u$ and $P_{N,\beta}$. This is a natural definition given that the logarithmic equivalence is the defining scale of large deviation theory in general, and thermodynamic LDPs in particular. Our final result shows that this definition is also related to the concavity of the entropy, which means that it relates physically to all the definitions of equivalence studied before. The probabilistic interpretation of this new result, which can actually be extended to general measures beyond the microcanonical and canonical ensembles, is discussed in the next subsection.

\begin{theorem}[Measure equivalence II]
\label{thmmeq2}
Assume A1-A4. Then
\begin{enumerate}
\item[\emph{(a)}] \emph{Strict equivalence:} If $s$ is strictly concave at $u$, then measure equivalence holds in the specific action sense for all $\beta\in\partial s(u)$;
\item[\emph{(b)}] \emph{Nonequivalence:} If $s$ is nonconcave at $u$, then measure equivalence does not hold in the specific action sense for any $\beta\in\reals$;
\item[\emph{(c)}] \emph{Partial equivalence:} If $s$ is concave at $u$ but not strictly concave, then 
\be
\lim_{N\ra\infty} R_{N,\beta}^u(\om)=0
\label{eqmeaseq2}
\ee
$P_N^u$-almost everywhere for all $\beta\in\partial s(u)$, but the same limit is in general undefined with respect to $P_{N,\beta}$.
\end{enumerate}
\end{theorem}

\begin{proof}
Recall that $R_{N,\beta}^u$ is a random variable that depends on the two parameters $\beta\in\dom\varphi$ and $u\in\dom s$. From the explicit expression of the Radon-Nikodym derivative found in (\ref{eqrndt1}), we have in fact
\be
r_\beta^u(\om)=\lim_{N\ra\infty} \frac{1}{N}\ln \frac{dP_{N}^{u}}{dP_{N,\beta}}(\om)=
\left\{
\begin{array}{lll}
J_\beta(u) & & h_N(\om)\in du\\
-\infty & & \text{otherwise}.
\end{array}
\right.
\label{eqrndu1}
\ee
Thus the limit $r_{\beta}^u(\om)$ is also a random variable, and since it depends only on $h_N(\om)$, it inherits by the contraction principle the LDP of $h_N$ with respect to $P_N^u$ or $P_{N,\beta}$, which means that we can describe its concentration in terms of these LDPs. 

To prove the different cases of the theorem in a complete way, we will distinguish between the parameters $u$ and $\beta$ of $r_\beta^u$ and those of the microcanonical and canonical ensemble, which we denote instead by $u'\in\dom s$ and $\beta'\in\dom\varphi$, respectively. Thus, we want to study  the concentration of $r_\beta^u$ with respect to $P_n^{u'}$ and $P_{N,\beta'}$.

We begin with the microcanonical ensemble. Clearly, $h_N(\om)\in du'$ with probability 1 with respect to $P_N^{u'}$, so that
\be
r_\beta^u(\om)=
\left\{
\begin{array}{lll}
J_\beta(u) & & u'=u\\
-\infty & & \textrm{otherwise}
\end{array}
\right.
\ee 
for all $\om$ relative to $P_N^{u'}$. Moreover, we know from Proposition~\ref{lemu1} that $J_\beta(u)=0$ if and only if $\beta\in \partial s(u)$. Therefore, $r_\beta^u=0$ relative to $P_N^u$ if $\beta\in\partial s(u)$ and $r_\beta^u\neq 0$ otherwise, proving the microcanonical half of the theorem.

For the canonical ensemble, the concentration is more involved and must be treated following the three different cases considered:
\begin{enumerate}
\item[(a)] $s$ is strictly concave at $u$: In this case, we know by Proposition~\ref{lemu1} that $J_\beta(u)=0$ and $\cU_\beta=\{u\}$ for all $\beta\in\partial s(u)$. This means that $u$ is the unique equilibrium value of $h_N$ with respect to $P_{N,\beta}$, so that
\be
\lim_{N\ra\infty} P_{N,\beta}\{h_N\in du\}= 1.
\ee
Thus, although $r_\beta^u(\om)$ diverges for $\om$ such that $h_N\notin du$, these microstates have zero measure with respect to $P_{N,\beta}$, so that $r_\beta^u=0$ almost everywhere with respect to $P_{N,\beta}$. This also holds with respect to $P_{N,\beta'}$ if $\beta'\in\partial s(u)$ and $\beta\in\partial s(u)$ but $\beta'\neq \beta$ because in that case $h_N$ still concentrates on $u$ with respect to $P_{N,\beta'}$. However, if $\beta'\notin\partial s(u)$, then $h_N$ will not concentrate on $u$, implying $r_\beta^u(\om)=-\infty$.

\item[(b)] $s$ is nonconcave at $u$: In this case, we also know from Proposition~\ref{lemu1} that $u\notin\cU_\beta$ for all $\beta\in\reals$, so that $J_\beta(u)>0$ for all $\beta\in\reals$. This directly implies $r_\beta^u\neq 0$ with respect to $P_{N,\beta'}$ with any $\beta'\in\reals$ including $\beta'=\beta$. To be more precise, we must have in fact $r_\beta^u(\om)=-\infty$ almost surely with respect to $P_{N,\beta'}$ for any $\beta'\in\reals$, since $h_N$ does not concentrate on $u$ for any $\beta'\in\reals$.

\item[(c)] $s$ is non-strictly concave at $u$: In this case, Proposition~\ref{lemu1} implies that $J_\beta(u)=0$ for $\beta\in\partial s(u)$; however, although $u\in\cU_\beta$, $u$ is not the only element of $\cU_\beta$, which means that the concentration point of $h_N$ with respect to $P_{N,\beta'}$ with $\beta'=\beta$ is in general unknown: it can be $u$, in which case $r_\beta^u(\om)=0$ as (a), or it can be a different mean energy value, in which case $r_\beta^u(\om)=-\infty$ as in (b). 
\end{enumerate}
The indefinite result in (c) is a consequence again of the phase coexistence arising when $s(u)$ is non-strictly concave. If we make the additional assumption that $u\in\cU_\beta$ \emph{is} a concentration point of $h_N$ with respect to $P_{N,\beta}$, as discussed in Subsection~\ref{secdefeqmac}, then we recover measure equivalence as in case (a). Consequently, under this additional hypothesis, measure equivalence holds in the specific action sense if and only if $s(u)$ is concave (strictly or non-strictly) and, therefore, if and only if macrostate and thermodynamic equivalence holds.
\end{proof}

The roles of $P_{N,\beta}$ and $P_N^u$ can be reversed in all the results of this section to study the convergence of $dP_{N,\beta}/dP_N^u$ instead of $dP_{N}^u/dP_{N,\beta}$. Indeed, although the former Radon-Nikodym derivative diverges for some $\om\in\Lambda_N$ because $P_N^u \not\gg P_{N,\beta}$, these divergences happen to be exactly cancelled when measure equivalence holds because $P_{N,\beta}$ concentrates towards $P_N^u$ in the thermodynamic limit, which implies that these divergencies have zero measure in the canonical ensemble. From this, one can re-derive results similar to Theorems~\ref{thmmeq1} and \ref{thmmeq2} with $P_{N,\beta}$ and $P_N^u$ interchanged in the definitions of the relative entropy and Radon-Nikodym derivative.

Measure equivalence can also be derived directly from the macrostate level, in two different ways in fact, without having to calculate the Radon-Nikodym derivative, as done above. On the first hand, we can consider the specific action $R_{N,\beta}^u$ as a macrostate (it is a function of $\om$) and apply our results of Sec.~\ref{secmacro}. The LDP of this macrostate with respect to $P_N^u$ is trivial, while its LDP with respect to $P_{N,\beta}$ follows by contraction from the LDP of $h_N$ in the canonical ensemble, as mentioned in the proof of Theorem~\ref{thmmeq2}. With these LDPs, we can then apply Theorem~\ref{theoremellis2} to obtain Theorem~\ref{thmmeq2}, which clearly demonstrates that measure equivalence is directly related to macrostate equivalence.

On the other hand, we can consider the level-3 empirical process, mentioned before, and prove that the equilibrium points of this infinite-dimensional macrostate converge to the ensemble measures in the thermodynamic limit. This more abstract approach is followed in \cite{deuschel1991,stroock1991,stroock1991b,roelly1993,georgii1993,georgii1994,georgii1995} and is also used for proving the equivalence of Gibbs measures (or Gibbs random fields) and translationally invariant measures in the thermodynamic limit \cite{follmer1988b,follmer1988,georgii1988}.

\subsection{Asymptotic equipartition property}

The integral (\ref{eqmix1}) is a Laplace integral that concentrates in an exponential way as $N\ra\infty$ on the set $\cU_\beta$ of canonical equilibrium values of $h_N$, as explained before. In the particular case where $\cU_\beta$ is a singleton $\{u_\beta\}$, we can approximate this integral on the exponential scale to formally write
\be
P_{N,\beta}(d\om) \asymp P_N^{u_\beta}(d\om),
\ee
which recovers our definition of measure equivalence based on the specific action and the logarithmic equivalence. For a uniform prior $P_{N}(d\om)$ this means that, although $P_{N,\beta}(d\om)$ varies in general from one microstate to another according to their energy, most microstates with respect to $P_{N,\beta}$ are roughly equiprobable, as in the microcanonical ensemble, because most of these microstates have a constant energy $u_\beta$ with respect to $P_{N,\beta}$. 

In information theory, this equiprobability property of random sequences (here microstates) is called the \emph{asymptotic equipartition property} (AEP) and the set of sequences (viz., microstates) having this property is called the \emph{typical} or \emph{typicality set} \cite{cover1991}. In information theory, these sequences are those that contribute most to the entropy of a source because they appear in a typical way, whereas in statistical physics the corresponding microstates are those that contribute most to the thermodynamics and equilibrium behavior of a system in the thermodynamic limit. The number (or volume) of such microstates can be estimated as follows. Let $\Lambda_{N,u}$ denote the subset of microstates having a mean energy $h_N(\om)$ close to $u$, that is,
\be
\Lambda_{N,u}=\{\om\in\Lambda_N: h_N(\om)\in du\}.
\ee
Assuming strict equivalence, we have
\be
P_{N,\beta}\{h_N\in du_\beta\}=P_{N,\beta}\{\Lambda_{N,u_\beta}\}\asymp 1
\ee
for a unique $u_\beta$, which means $\Lambda_{N,u_\beta}$ is a typical set in the canonical ensemble. For $P_{N}$ uniform, we thus have that most microstates are such that $H_N(\om)=Nu_\beta+o(N)$ and
\be
P_{N,\beta}(d\om)\asymp \frac{e^{-\beta N u_\beta}}{Z_N(\beta)},
\ee
in the thermodynamic limit, which implies that the volume of these microstates must approximately be given by 
\be
|\Lambda_{N,u_\beta}|\asymp e^{\beta N u_\beta} Z_N(\beta).
\ee
This form of AEP follows from our results for any $N$-particle system satisfying assumptions A1-A2, that is, any system with a well-defined thermodynamic-limit free energy and entropy.

\section{Other ensembles}
\label{secother}

As mentioned in the introduction, our discussion of ensemble equivalence centered on the canonical and microcanonical ensembles to be specific and to simplify the notations. In this section, we briefly discuss how these results are generalized to ensembles other than the canonical and microcanonical. By way of example, we start  with the equivalence of the canonical and grand-canonical ensembles, used for example to describe the liquid-gas transition, and then point out how more general dual ensembles can be treated following the results of \cite{ellis2000}. We discuss finally the case of ensembles defined on random paths of stochastic processes rather than static (spatial) configurations.

\subsection{Canonical and grand-canonical ensembles}

Denote by $H_V$ the energy of a system with volume $V$ and by $N_V$ its particle number. The grand-canonical ensemble associated with this system is defined by the probability measure
\be
P_{V,\beta,\mu}(d\om)=\frac{e^{-\beta (H_V(\om)-\mu N_V(\om))}}{Z_V(\beta,\mu)}P_V(d\om),
\ee
where
\be
Z_V(\beta,\mu)=\int_{\Lambda_V} e^{-\beta (H_V(\om)-\mu N_V(\om))}\, P_V(d\om)
\ee
is the \emph{grand-canonical partition function} and $P_V(d\om)$ is the prior measure on the space $\Lambda_V$ of microstates at volume $V$. This ensembles extends, as is well known, the canonical ensemble by allowing fluctuations of the particle number $N_V(\om)$ in a system of fixed volume $V$. In terms of the particle density $r_{V}(\om)=N_V(\om)/V$, the canonical ensemble with fixed density $r_V=\rho$ is then defined as
\be
P_{V,\beta}^\rho(d\om)= \frac{e^{-\beta H_V(\om)}}{W_V^\rho(\beta)}\, \id_{d\rho} \big(r_N(\om)\big) \, P_V(d\om),
\ee
where $Z_V^\rho(\beta)$ is a normalization factor given by
\be
Z_{V}^\rho(\beta)=\int_{\Lambda_N} e^{-\beta H_V(\om)}\, \id_{d\rho} \big(r_N(\om)\big)\, P_V(d\om)=E_{P_V}\left[e^{-\beta H_V}\, \id_{d\rho} (r_N)\right]
\ee
and, as before, $d\rho$ is some infinitesimal interval centered at $\rho$. The superscripts and subscripts in these expressions follow the notations of \cite{ellis2000} and denote either a microcanonical-like constraint (superscript $\rho$) or a canonical-like exponential (subscript $\mu$) involving a Lagrange parameter conjugated to the constraint.

Comparing these ensembles with the definitions of the original canonical and microcanonical ensembles, it is easy to see that the grand-canonical ensemble conditioned on a fixed value of the particle density $r_V=\rho$ is equivalent to the canonical ensemble, which means that the former is a probabilistic mixture of the latter, with $r_N$ playing the role of the `mixing' random variable. Following our discussion of macrostate equivalence, the probability measure of $r_N$ that determines this mixture is the one obtained in the non-constrained ensemble, that is, the grand-canonical ensemble. Assuming that this probability measure satisfies the LDP,
\be
P_{V,\beta,\mu}\{r_N\in d\rho\}\asymp e^{-V J_{\beta,\mu}(\rho)}
\ee
in the thermodynamic limit $V\ra\infty$ with $\rho=r_V/V$ constant, we find from the probabilistic mixture that the rate function $J_{\beta,\mu}(\rho)$ is given by 
\be
J_{\beta,\mu}(\rho)=-\beta\mu\rho-s_\beta(\rho)-\varphi(\beta,\mu)
\label{eqratefctgc1}
\ee
where
\be
\varphi(\beta,\mu)=\lim_{V\ra\infty} -\frac{1}{V}\ln Z_V(\beta,\mu)
\ee
is the grand-canonical free energy, or grand potential, and 
\be
s_\beta(\rho) =\lim_{V\ra\infty}\frac{1}{V}\ln Z_V^\rho(\beta)
\label{eqsrho1}
\ee
is the thermodynamic potential associated with the canonical ensemble with fixed Lagrange parameter $\beta$ and fixed constraint $\rho$. The grand-canonical potential $\varphi(\beta,\mu)$ obviously plays the role of $\varphi(\beta)$ while $s_\beta(\rho)$ takes the role of $s(u)$. Therefore, what determines the equivalence of the grand-canonical and canonical ensemble, with respect to $r_N$, is the concavity of $s_\beta(\rho)$ as a function of $\rho$. In other words, all our results involving $s(u)$ generalize to these ensembles by considering $s_\beta(\rho)$ instead.

To see this more clearly, rewrite the grand-canonical and canonical measures as
\be
P_{V,\beta,\mu}(d\om)=\frac{e^{-\gamma V r_V(\om)}}{Z_V(\beta,\mu)}Q_{V,\beta}(d\om),
\ee
and
\be
P_{V,\beta}^\rho(d\om)=\frac{\id_{d\rho}\big(r_N(\om)\big)}{W_V^\rho(\beta)}\,  Q_{V,\beta}(d\om),
\ee
respectively, by defining $\gamma=-\beta\mu$ and the positive but non-normalized measure
\be
Q_{V,\beta}(d\om)=e^{-\beta H_V(\om)} P_V(d\om).
\ee
Then these ensembles take the same form as the canonical and microcanonical ensembles, respectively, but with the prior measure $P_N$ replaced by $Q_{V,\beta}$. Moreover, $h_N$ is replaced by $r_V$ while $N$ is replaced by $V$. As a result, the entropy function $s(u)$ defined in (\ref{eqsu1}) which determines equivalence between the canonical and microcanonical ensembles must now be defined for $r_N$ with respect to $Q_{V,\beta}$, which leads us to $s_\beta(\rho)$ as defined in (\ref{eqsrho1}).

For applications of these ideas to the case of two constraints involving the energy and magnetization, see \cite{campa2007,kastner2010,kastner2010b,olivier2014}; for an application to the zero-range process with a single particle density constraint, see \cite{grosskinsky2008}.

\subsection{Mixed ensembles}

Ensembles involving more than one constraints can be treated along the lines just discussed or, more completely, by following Sec.~5 of Ellis, Haven and Turkington \cite{ellis2000} who refer to these ensembles as `mixed ensembles'. Here, we briefly summarize the changes that need to be taken into account, following the notations of \cite{ellis2000}. In terms of definitions, the changes are as follows:
\begin{itemize}
\item Write all the conserved quantities $h_{N,1},\ldots,h_{N,\sigma}$ considered in the model as a vector $h_N=(h_{N,1},\ldots,h_{N,\sigma})$, referred to as the \emph{generalized Hamiltonian}.

\item Denote the quantities to be treated canonically as $h_N^1$ and those to be treated microcanonically (as constraints) as $h_N^2$. Then write $h_N=(h_N^1,h_N^2)$.

\item Associate a vector $\beta=(\beta_1,\ldots,\beta_\sigma)$ of Lagrange parameters to $h_N$ and denote the restriction of that vector associated with the canonical part $h_N^1$ by $\beta^1$.

\item Define the full canonical ensemble for $h_N$ as
\be
P_{N,\beta}(d\om)=\frac{e^{- N \langle \beta, h_N(\om)\rangle }}{Z_N(\beta)} P_N(d\om)
\ee
where $\langle\beta, h_N\rangle=\sum_{i=1}^\sigma \beta_i h_{N,i}$ is the normal scalar product.

\item Define the \emph{mixed ensemble} with $h_N^1$ treated canonically and $h_N^2$ treated microcanonically as
\be
P_{N,\beta_1}^{u^2}(d\om)=\frac{e^{-N\langle \beta_1, h_N^1(\om)\rangle}}{Z_{N}^{u^2}(\beta_1)}\, \id_{du^2}\big(h_N^2(\om)\big)\, P_N(d\om)
\ee
\end{itemize}

The equivalence of the canonical and mixed ensembles is determined using the same results as before with the following changes:
\begin{itemize}
\item The real parameter $\beta$ is now a vector in $\reals^\sigma$.
\item The real parameter $u$ is replaced by the vector $u^2$.
\item $\varphi(\beta)$ is still defined from $Z_N(\beta)$ with $\beta$ now a vector.
\item $s(u)$ is replaced by the thermodynamic potential $s_{\beta_1}(u^2)$ of the mixed ensemble:
\be
s_{\beta_1}(u^2)=\lim_{N\ra\infty} \frac{1}{N}\ln Z_N^{u^2}(\beta_1).
\ee

\item The product $\beta u$ in the Legendre-Fenchel transform is replaced by the scalar product $\langle\beta_2, u^2\rangle$.
\item Supporting lines must be replaced by supporting planes or hyperplanes; see \cite{ellis2000}.
\item Concave points of vector functions have supporting hyperplanes in their domain except possibly at relative boundary points, that is, points on the boundary of the relative interior of their domain; see also \cite{ellis2000}.
\end{itemize}

The first change concerning $s(u)$ should be clear from our discussion of the grand-canonical and canonical case; for more details, see Sec.~5 of \cite{ellis2000}

\subsection{Nonequilibrium ensembles}

The physical interpretation of $P_{N,\beta}$ and $P_N^u$ is not important for establishing their equivalence in the large $N$ limit. Clearly, this equivalence is a general relation between a measure conditioned on some event or constraint and a measure obtained by replacing this conditioning with an exponential factor involving a Lagrange parameter dual to the constraint. Mathematically, we say that equivalence is between a conditioning and a tilting of the same measure, in a scaling limit that depends on the nature of the objects or structures on which these measures are defined.

This general view of ensemble equivalence is potentially useful for replacing constrained (Monte Carlo) sampling schemes, arising for example in rare event simulations \cite{asmussen2007} and the sampling of random graphs \cite{squartini2015}, by modified sampling schemes based on exponentially-tilted distributions. It can also be used to establish the equivalence of microcanonical and canonical \emph{path ensembles} that are useful for describing the properties of nonequilibrium systems.

To illustrate this case, consider the probability measure $P_T(d\om)$ defined on the space $\Lambda_T$ of random paths $\om=\{\om_t\}_{t=0}^T$ of a continuous-time process evolving over a time interval $[0,T]$. This probability plays the role of the prior $P_N$. For a macrostate or \emph{observable} $A_T$, which is a functional of $\om$, it is natural to define a microcanonical path ensemble as
\be
P_T^a(d\om)=P_T\{d\om|A_T\in da\},
\ee
to describe the subset of paths of the process leading to a fluctuation $A_{T}=a$. The corresponding canonical path ensemble is
\be
P_{T,k}(d\om)=\frac{e^{kTA_T(\om)}}{W_T(k)}P_T(d\om)
\ee
where $k\in\reals$ and
\be
W_T(k)=E_{P_T}[e^{kTA_T}].
\ee
From these definitions, we see that the parameter $k$ plays the role of (minus) an inverse temperature and that the time $T$ plays the role of the particle number $N$, so that the thermodynamic limit is now $T\ra\infty$ with $A_T$ finite. In this limit, all our equivalence results holds for $P_T^a$ and $P_{T,k}$ under assumptions similar to A1-A4. In particular, assuming that $A_T$ satisfies the LDP with respect to $P_T$ with rate function $I(a)$, then $P_T^a$ and $P_{T,k}$ are equivalent in the specific action sense if $I$ is convex at $a$.\footnote{Convexity is used instead of concavity because $I$ is defined as a rate function rather than an entropy function.} In this case, we also have macrostate equivalence, which means that $P_T^a$ and $P_{T,k}$ lead to the same stationary or ergodic values of observables. 

The equivalence of these path ensembles was discussed recently in \cite{jack2010b,chetrite2013,chetrite2014} for general Markov processes. An interesting open problem is to find examples of stochastic processes and observables characterized by nonconvex rate functions for which the microcanonical and canonical path ensembles are not equivalent. For applications of these ensembles in the context of sheared fluids, glasses, and other nonequilibrium systems, see \cite{evans2005a} and the review \cite{chandler2010}.

\appendix

\section{Concavity of the entropy}
\label{appconc}

Let $s:\reals\ra\reals\cup\{-\infty\}$ be a real function with domain $\dom s$, and consider the inequality
\be
s(v)\leq s(u)+\beta(v-u),\quad v\in\reals.
\ee
The set of all $\beta\in\reals$ for which this inequality is satisfied is called the \emph{subdifferential set} or simply the \emph{subdifferential} of $s$ at $u$ and is denoted by $\partial s(u)$.\footnote{The term `superdifferentials' should be used for concave functions, but we will keep to the more common term `subdifferentials'.} The interpretation of this inequality is shown in Fig.~\ref{fignoneq1}: if it is possible to draw a line passing through the graph of $s(u)$ which is everywhere above $s$, then $\partial s(u)\neq \emptyset$. In this case, we also say that $s$ admits a \emph{supporting line} at $u$, which is unique if $s$ is differentiable at $u$. If $\partial s(u)=\emptyset$, then $s$ admits no supporting line at $u$. 
 
It is easy to see geometrically that nonconcave points of $s$ do not admit supporting lines, while concave points have supporting lines, except possibly if they lie on the boundary of $\dom s$; see Sec.~24 of \cite{rockafellar1970} or Appendix A of \cite{costeniuc2005}. The reason for possibly excluding boundary points arises because $s(u)$ may have diverging `slopes' where $\partial s(u)$ is not defined, as in the following example adapted from \cite[p.~215]{rockafellar1970}:
\be
s(u)=\left\{
\begin{array}{lll}
\sqrt{1-x^2} & & |x|\leq 1\\
-\infty & & \textrm{otherwise}.
\end{array}
\right.
\ee 
In this case, $s(u)=s^{**}(u)$ for all $u\in\dom s=[-1,1]$, so that $s$ is a concave function, but it has supporting lines only over $(-1,1)=\textrm{int}(\dom s)$, since $s'(u)$ diverges as $u\ra\pm 1$ from within its domain. All cases of concave points with no supporting lines are of this type, since it can be proved in $\reals$ that
\be
\textrm{int}(\dom s)\subseteq \dom \partial s\subseteq \dom s;
\ee
see again Sec.~24 of \cite{rockafellar1970} or Appendix A of \cite{costeniuc2005}. 

With this proviso on boundary points, $s$ is often defined to be \emph{strictly concave} at $u$ if it admits supporting lines at $u$ that do not touch other points of its graph. If $s$ has a supporting line at $u$ touching other points of its graph, then $s$ is said to be \emph{non-strictly concave} at $s$. Finally, if $s$ admits no supporting line at $u$, then $s$ is said to be \emph{nonconcave} at $u$. These definitions are also illustrated in Fig.~\ref{fignoneq1}. For generalizations of these definitions to $\reals^d$ in terms of \emph{supporting hyperplanes}, see \cite{rockafellar1970} and Appendix A of \cite{costeniuc2005}.

\section{Varadhan's Theorem and the Laplace principle}
\label{appvar}

We recall in this section two important results about Laplace approximations of exponential integrals in general spaces. In the following, $\{a_n\}_{n=1}^\infty$ is an increasing sequence such that $a_n\nearrow\infty$ when $n\ra\infty$. Moreover, $\{P_n\}_{n=1}^\infty$ is a sequence of probability measures defined on a (Polish) space $\cX$. In this paper, $N$ takes the role of $a_n$ and $n$.

\begin{theorem}[Varadhan, 1966 \cite{varadhan1966}]
Assume that $P_n(dx)$ satisfies the LDP with speed $a_n$ and rate function $I$ on $\cX$. Let $F$ be a continuous function.
\begin{enumerate}
\item[\emph{(a)}] \emph{(Bounded case)} Assume that $\sup_x F(x)<\infty$. Then
\be
\lim_{n\ra\infty} \frac{1}{a_n} \ln \int_{\cX} e^{a_n F(x)} P_n(dx) =\sup_{x\in\cX} \{F(x)-I(x)\}<\infty.
\label{eqvar1}
\ee

\item[\emph{(b)}] \emph{(Unbounded case)} Assume that $F$ satisfies
\be
\lim_{L\ra\infty} \lim_{n\ra\infty} \frac{1}{a_n}\ln \int_{\{F\geq L\}} e^{a_n F(x)} P_n(dx)=-\infty
\ee
Then the result of \emph{(a)} holds and is finite. In particular, if $F$ is bounded above on the support of $P_n$, then \emph{(a)} holds. 
\end{enumerate}
\end{theorem}

For a proof of this result, see the Appendix B of \cite{ellis1985} or Theorem 4.3.1 in \cite{dembo1998}. For historical notes on this result, see Sec.~3.7 of \cite{touchette2009}

Consider now the \emph{exponentially tilted} probability measure
\be
P_{n,F}(dx) = \frac{e^{a_n F(x)} P_n(dx)}{W_{n,F}},
\ee
where
\be
W_{n,F}=\int_{\cX} e^{a_nF(x)} P_n(dx)=E_P[e^{a_nF(X)}].
\ee
This is also known as the \emph{exponential family} or \emph{Esscher transform} of $P_n$.

\begin{theorem}[LDP for tilted measures]
\label{thmtvar1}
Assume that $W_{n,F}<\infty$. Then $P_{n,F}$ satisfies the LDP with speed $a_n$ and rate function
\be
I_{F} (x)=I(x)-F(x)+\lambda_F,
\ee
where
\be
\lambda_F=\lim_{n\ra\infty}\frac{1}{a_n}\ln W_{n,F}.
\ee
\end{theorem}

 A proof of this result can be found in Theorem 11.7.2 of \cite{ellis1985} or by combining Proposition 3.4 and Theorem 9.1 of \cite{ellis1995}. A thermodynamic version of this result also appears in Theorem 4.1 of \cite{lewis1995}.

\begin{acknowledgments}
I would like to thank many colleagues who have provided useful ideas, comments, and support during the last 12 years that I worked on long-range systems and nonequivalent ensembles: Julien Barr\'e, Freddy Bouchet, Raphael Chetrite, Thierry Dauxois, Rosemary J.~Harris, Michael Kastner, Cesare Nardini, and Stefano Ruffo. I especially want to thank Richard S. Ellis for introducing me to many gems and subtleties of large deviations. The present paper owes much to his work.
\end{acknowledgments}

\bibliography{masterbib}

\end{document}